\documentclass[11pt,letterpaper]{article}

\usepackage[margin=1in]{geometry}

\usepackage{theorem,latexsym,graphicx,amssymb}
\usepackage{amsmath,enumerate}

\usepackage[numbers]{natbib}

\usepackage{subfigure}
\usepackage{float}
\usepackage{wrapfig}
\usepackage{psfig}
\usepackage{epsfig}
\usepackage{xspace}
\usepackage{paralist}
\usepackage{times}
\usepackage[compact]{titlesec}
\usepackage{enumerate}
\usepackage{cases}
\usepackage{caption}
\usepackage{mathtools}

\usepackage{algorithm}
\usepackage{algorithmicx}
\usepackage{algpseudocode}
\usepackage{authblk}
\usepackage{tikz}

\usepackage{color}
\definecolor{Darkblue}{rgb}{0,0,0.4}
\definecolor{Brown}{cmyk}{0,0.81,1.,0.60}
\definecolor{Purple}{cmyk}{0.45,0.86,0,0}
\usepackage[breaklinks]{hyperref}
\hypersetup{colorlinks=true,
            citebordercolor={.6 .6 .6},linkbordercolor={.6 .6 .6},%
citecolor=black,urlcolor=black,linkcolor=black,pagecolor=black}
\newcommand{\lref}[2][]{\hyperref[#2]{#1~\ref*{#2}}}


\usepackage{multicol}

\newenvironment{proof}{{\bf Proof:  }}{\hfill\rule{2mm}{2mm}}

\numberwithin{figure}{section}
\numberwithin{equation}{section}

\newtheorem{definition}{Definition}[section]

\newtheorem{corollary}{Corollary}[section]
\newtheorem{theorem}{Theorem}[section]
\newtheorem{lemma}{Lemma}[section]

\newtheorem{theorem*}{Theorem}
\newtheorem{corollary*}{Corollary}

\newtheorem{proposition}{\hskip\parindent Proposition}[section]

\renewcommand{\theequation}{\thesection.\arabic{equation}}
\renewcommand{\thefigure}{\thesection.\arabic{figure}}

\providecommand{\Appendix}{}
\renewcommand{\Appendix}[2][?]{%
        \refstepcounter{section}%
        \vspace{\parskip}%
        {\flushright\large\bfseries\appendixname\ \thesection: #1}%
        \vspace{\baselineskip}%
}

\renewcommand{\appendix}{%
        \newpage
        \renewcommand{\section}{\secdef\Appendix\Appendix}%
        \renewcommand{\thesection}{\Alph{section}}%
        \setcounter{section}{0}%
}

\newcounter{note}[section]

\newcommand{\ignore}[1]{}

\newcommand{\SREV}{ \ensuremath{\mathsf{SRev}}}

\newcommand{\I}{ \ensuremath{\mathsf{I}}}
\newcommand{\COR}{ \ensuremath{\mathsf{Cor}}}
\newcommand{\BREV}{ \ensuremath{\mathsf{BRev}}}

\newcommand{\REV}{ \ensuremath{\mathsf{Rev}}}
\newcommand{\VAL}{ \ensuremath{\mathsf{Val}}}
\newcommand{\VAR}{ \ensuremath{\mathsf{Var}}}
\newcommand{\COVAR}{ \ensuremath{\mathsf{Covar}}}

\newcommand{\itemsbought}{S_{V}}

\newcommand{\convexdistribution}{linear distribution }

{\makeatletter
 \gdef\xxxmark{%
   \expandafter\ifx\csname @mpargs\endcsname\relax 
     \expandafter\ifx\csname @captype\endcsname\relax 
       \marginpar{xxx}
     \else
       xxx 
     \fi
   \else
     xxx 
   \fi}
 \gdef\xxx{\@ifnextchar[\xxx@lab\xxx@nolab}
 \long\gdef\xxx@lab[#1]#2{{\bf [\xxxmark #2 ---{\sc #1}]}}
 \long\gdef\xxx@nolab#1{{\bf [\xxxmark #1]}}
}

\title{
Revenue Maximization for Selling Multiple Correlated Items\footnote{Supported in part by NSF CAREER award 1053605, NSF grant CCF-1161626, ONR YIP award N000141110662, and a DARPA/AFOSR grant FA9550-12-1-0423.}}

\author[1]{MohammadHossein Bateni}
\author[2]{Sina Dehghani}
\author[2]{MohammadTaghi Hajiaghayi}
\author[2]{Saeed Seddighin}
\affil[1]{Google Research \url{bateni@google.com}}
\affil[2]{Department of Computer Science, Univeristy of Maryland \url{dehghani, hajiagha, sseddigh@umd.edu}}

\date{}

\begin{document}

\begin{titlepage}

\maketitle

\begin{abstract}
We study the problem of selling $n$ items to a single buyer with an additive valuation function. We consider the valuation of the items to be correlated, i.e., desirabilities of the buyer for the items are not drawn independently. Ideally, the goal is to design a mechanism to maximize the revenue. However, it has been shown that a revenue optimal mechanism might be very complicated and as a result inapplicable to real-world auctions. Therefore, our focus is on designing a simple mechanism that achieves a constant fraction of the optimal revenue. 
\citeauthor{babaioff2014simple} \cite{babaioff2014simple} propose a simple mechanism that achieves a constant fraction of the optimal revenue for independent setting with a single additive buyer. However, they leave the following problem as an open question: \textit{``Is there a simple, approximately optimal mechanism for a single additive buyer whose value for $n$ items is sampled from a common base-value distribution?"} 
 \citeauthor{babaioff2014simple} show a constant approximation factor of the optimal revenue can be achieved by either selling the items separately or as a whole bundle in the independent setting. We show a similar result for the correlated setting when the desirabilities of the buyer are drawn from a common base-value distribution. It is worth mentioning that the core decomposition lemma which is mainly the heart of the proofs for efficiency of the mechanisms does not hold for correlated settings. Therefore we propose a modified version of this lemma which is applicable to the correlated settings as well. Although we apply this technique to show the proposed mechanism can guarantee a constant fraction of the optimal revenue in a very weak correlation, this method alone can not directly show the efficiency of the mechanism in stronger correlations. Therefore, via a combinatorial approach we reduce the problem to an auction with a weak correlation to which the core decomposition technique is applicable. In addition, we introduce a generalized model of correlation for items and show the proposed mechanism  achieves an $O(\log k)$ approximation factor of the optimal revenue in that setting.
 
\end{abstract}
\thispagestyle{empty}
\end{titlepage}

\section{Introduction}
Suppose an auctioneer wants to sell $n$ items to a single buyer. The buyer's valuation for a particular item comes from a known distribution, and the his values are assumed to be additive (i.e., value of a set of items for the buyer is equal to the summation of the values of the items in the set). The buyer is considered to be strategic, that is, he is trying to maximize $v(S)-p(S)$, where $S$ is the set of purchased items, $v(S)$ is the value of these items to the buyer and $p(S)$ is the price of the set. Knowing that the valuation of the buyer for item $j$ is drawn from a given distribution $D_{j}$, what is a revenue optimal mechanism for the auctioneer to sell the items? \citeauthor{myerson1981optimal} \cite{myerson1981optimal} solves the problem for a very simple case where we only have a single item and a single buyer. He shows that in this special case the optimal mechanism is to set a fixed reserved price for the item. Despite the simplicity of the revenue optimal mechanism for selling a single item, this problem becomes quite complicated when it comes to selling two items even when we have only one buyer. Hart and Reny \cite{hart2012maximal} show an optimal mechanism for selling two independent items is much more subtle and may involve randomization.

Though there are several attempts to characterize the properties of a revenue optimal mechanism of an auction, most approaches seem to be too complex and as a result impractical to real-world auctions \cite{alaei2012bayesian,alaei2013simple,bhalgat2013optimal,bhattacharya2010budget,cai2012algorithmic,cai2012optimal,cai2013reducing,daskalakis2014complexity,chawla2010multi,hart2012approximate,kleinberg2012matroid}. Therefore, a new line of investigation is to design simple mechanisms that are approximately optimal. In a recent work of Babaioff, Immorlica, Lucier, and Weinberg \cite{babaioff2014simple}, it is shown that we can achieve a constant factor approximation of the optimal revenue by selling items either separately or as a whole bundle in the independent setting. However, they leave the following important problem as an open question:\begin{itemize}\item
\textit{``\textbf{Open Problem 3}. Is there a simple, approximately optimal mechanism for a single additive buyer whose value for $n$ items is sampled from a common base-value distribution? What about other models of limited correlation?}" \end{itemize}

\citeauthor{hart2013menu} \cite{hart2013menu} show there are instances with correlated valuations in which neither selling items separately nor as a whole bundle can achieve any approximation of the optimal revenue. This holds, even when we have only two times. Therefore, it is essential to consider limited models of correlation for this problem. As an example, \citeauthor{babaioff2014simple} propose to study common base-value distributions. This model has also been considered by Chawla, Malec, and Sivan \cite{chawla2010power} to study optimal mechanisms for selling multiple items in a unit-demand setting.

In this work we study the problem for the case of correlated valuation functions and answer the above open question. In addition we also introduce a generalized model of correlation between items. Suppose we have a set of items and want to sell them to a single buyer. The buyer has a set of features in his mind and considers a value for each feature which is randomly drawn from a known distribution. Furthermore, the buyer formulates his desirability for each item as a linear combination of the values of the features. More precisely, the buyer has $l$ distributions $F_1,F_2,\ldots,F_l$ and an $l \times n$ matrix $M$ (which are known in advance) such that the value of feature $i$, denoted by $f_i$, is drawn from $F_i$ and the value of item $j$ is calculated by $V_f \cdot M_j$ where $V_f = \langle f_1,f_2,\ldots,f_l \rangle$ and $M_j$ is the $j$-th row of matrix $M$.

This model captures the behavior of the auctions  especially when the items have different features that are of different value to the buyers. Note that every common base-value distribution is a special case of this general correlation  where we have $n+1$ features $F_1,F_2,\ldots,F_n,B$ and the value of item $j$ is determined by $v_j+b$ where $v_j$ is drawn from $F_j$ and $b$ is equal for all items which is drawn from distribution $B$.

\section{Related Work}
As mentioned earlier, the problem originates from the seminal work of Myerson \cite{myerson1981optimal} in 1981 which characterizes a revenue optimal mechanism for selling a single item to a single buyer. This result was important in the sense that it was simple and practical while promising the maximum possible revenue. 
In contrast to this result, it is known that designing an optimal mechanism is much harder for the case of multiple items. There has been some efforts to find a revenue optimal mechanism for selling two heterogeneous items \cite{pavlov2011optimal} but, unfortunately, so far too little is known about the problem even for this simple case.

Hardness of this problem is even more highlighted when Hart and Reny \cite{hart2012maximal} observed randomization is necessary for the case of multiple items. This reveals the fact that even if we knew how to design an optimal mechanism for selling multiple items, it would be almost impossible to implement the optimal strategy in a real-world auction. Therefore, so far studies are focused on finding simple and approximately optimal mechanisms.

Speaking of simple mechanisms, it is very natural to think of selling items separately or as a whole bundle. The former mechanism is denoted by $\SREV$ and the latter is referred to by $\BREV$. Hart and Nissan \cite{hart2012approximate} show $\SREV$ mechanism achieves at least an $\Omega(1/\log^2n)$ approximation of the optimal revenue in the independent setting and $\BREV$ mechanism yields at least an $\Omega(1/\log n)$ approximation for the case of identically independent distributions. Later on, this result was improved by the work of Li and Yao, that prove an $\Omega(1/\log n)$ approximation factor for $\SREV$ and a constant factor approximation  for $\BREV$ for identically independent distributions \cite{li2013revenue}. These bounds are tight up to a constant factor. Moreover, it is shown $\BREV$ can be $\theta(n)$ times worse than the revenue of an optimal mechanism in the independent setting. Therefore in order to achieve a constant factor approximation mechanism we should think of more non-trivial strategies.

The seminal work of \citeauthor{babaioff2014simple} \cite{babaioff2014simple} shows despite the fact that both strategies $\SREV$ and $\BREV$ may separately result in a bad approximation factor, $\max\{\SREV,\BREV\}$ always has a revenue at least $\frac{1}{6}$ of an optimal mechanism. They also show we can determine which of these strategies has more revenue in polynomial time which yields a deterministic simple mechanism that can be implemented in polynomial time. However, there has been no significant progress in the case of correlated items, as \cite{babaioff2014simple} leave it as an open question.

In addition to this, they posed two more questions which became the subject of further studies. In the first question, they ask if there exists a simple mechanism which is approximately optimal in the case of multiple additive buyers? This question is answered by \citeauthor{yao2014n} \cite{yao2014n} via proposing a reduction from $k$-item $n$-bidder auctions to $k$-item auctions. They show, as a result of their reduction, a deterministic mechanism achieves a constant fraction of the optimal revenue by any randomized mechanism. In the second question, they ask if the same result can be proved for a mechanism with a single buyer whose valuation is $k$-demand? This question is also answered by a recent work of \citeauthor{rubinstein2015simple} \cite{rubinstein2015simple} which presents a positive result. They show the same mechanism that either sells the items separately or as a whole bundle, achieves a constant fraction of the optimal revenue even in the sub-additive setting with independent valuations. They, too, use the core decomposition technique as their main approach. Their work is very similar in spirit to ours since we both show the same mechanism is approximately optimal in different settings.

Another line of research investigated optimal mechanism for selling $n$ items to a single unit-demand buyer. \citeauthor{briest2010pricing} \cite{briest2010pricing} show how complex the optimal strategies can become by proving that the gap between the revenue of deterministic mechanisms and that of non-deterministic mechanisms can be unbounded even when we have a constant number of items with correlated values. This highlights the fact that when it comes to general correlations, there is not much that can be achieved by deterministic mechanisms. However, \citeauthor{chawla2010power} \cite{chawla2010power} study the problem with a mild correlation known as the common base-value correlation and present positive results for deterministic mechanisms in this case.
\section{Results and Techniques}
We study the mechanism design for selling $n$ items to a single buyer with additive valuation function when desirabilities of each buyer for items are correlated. The main result of the paper is $\max\{\SREV,\BREV\}$, that is, the revenue we get by the better of selling items separately or as a whole bundle achieves a constant approximation of the optimal revenue when we have only one buyer and the distribution of valuations for this buyer is a common base-value distribution. This problem was left open in  \cite{babaioff2014simple}. Our method for proving the effectiveness of the proposed mechanism is consisted of two parts. In the first part, we consider a very weak correlation between the items, which we call semi-independent correlation, and show the same mechanism achieves a constant fraction of the optimal revenue in this setting. To this end, we use the core decomposition technique which has been used by several similar works \cite{li2013revenue,babaioff2014simple,rubinstein2015simple}. The second part, however, is based on a combinatorial reduction which reduces the problem to an auction with a semi-independent valuation function.
\begin{theorem}
For an auction with one seller, one buyer, and a common base-value distribution of valuations we have
$\max\{\SREV(D),\BREV(D)\} \geq \frac{1}{12} \times \REV(D).$
\end{theorem}
Furthermore, we consider a natural model of correlation in which the buyer has a number of features and scores each item based on these features. The valuation of each feature for the buyer is realized from a given distributions which is known in advance. The value of each item to the buyer is then determined by a linear formula in terms of the values of the features. This can also be seen as a generalization of the common base-value correlation since a common base-value correlation can be though of as a linear correlation with $n+1$ features. We show that if all of the features have the same distribution then $\max\{\SREV(D),\BREV(D)\}$ is at least a $\frac{1}{O(\log k)}$ fraction of $\REV(D)$ where $k$ is the maximum number of features that determine the value of each item.
\begin{theorem}
In an auction with one seller, one buyer, and a linear correlation with i.i.d distribution of valuations for the features $\max\{\SREV,\BREV\} \geq O(\frac{\REV}{\log k})$ where the value of each item depends on at most $k$ features.
\end{theorem}

Our approach is as follows: First we study the problem in a setting which we call \textit{semi-independent}. In this setting, the valuation of the items are realized independently, but each item can have many copies with the same value. More precisely, each pair of items are either similar or different. In the former case, they have the same value for the buyer in each realization whereas in the latter case they have independent valuations. 

Inspired by \cite{babaioff2014simple}, we show $\max\{\SREV(D),\BREV(D)\} \geq \frac{\REV(D)}{6}$ for every semi-independent distribution $D$. To do so, we first modify the core decomposition lemma to make it applicable to the correlated settings. Next, we apply this lemma to the problem and prove $\max\{\SREV(D),\BREV(D)\}$ achieves a constant fraction of the optimal revenue.

Given $\max\{\SREV(D),\BREV(D)\}$ is optimal up to a constant factor in the semi-independent setting, we analyze the behavior of $\max\{\SREV,\BREV\}$ in each of the settings by creating another auction in which each item of the original auction is split into several items and the distributions are semi-independent. We show that the maximum achievable revenue in the secondary auction is no less than the optimal revenue of the original auction and also selling all items together has the same revenue in both auctions. Finally, we  bound the revenue of $\SREV$ in the original auction by a fraction of the revenue that $\SREV$ achieves in the new auction and by putting all inequalities together we prove an approximation factor for $\max\{\SREV,\BREV\}$.
 In contrast to the prior methods for analyzing the efficiency of mechanism, our approach in this part is purely combinatorial.

Although the main contribution of the paper is analyzing $\max\{\SREV,\BREV\}$ in common base-value and linear correlations, we show the following as auxiliary lemmas which might be of independent interest.
\begin{itemize}
	\item One could consider a variation of independent setting, wherein each item has a number of copies and the value of all copies of an item to the buyer is always the same. We show in this setting $\max\{\SREV,\BREV\}$ is still a constant fraction of $\REV$.
	\item A natural generalization of i.i.d settings, is a setting in which the distributions of valuations are not exactly the same, but are the same up to scaling. We show, in the independent setting with such valuation functions $\BREV$ is at least an $O(\frac{1}{\log n})$ fraction of $\REV$. 
\end{itemize}
\section{Preliminaries}
Throughout this paper we study the optimal mechanisms for selling $n$ items to a risk-neutral, quasi-linear buyer. The items are considered to be indivisible and not necessarily identical i.e. the buyer can have different distributions of desirabilities for different items. In our setting, distributions are denoted by $D = \langle D_1, D_2, \ldots, D_n \rangle$ where $D_j$ is the distribution for item $j$. Moreover, the buyer has a valuation vector $V = \langle v_1,v_2,\ldots,v_n \rangle$ which is randomly drawn from $D$ specifying the values he has for the items. Note that, values may be correlated.\\
Once a mechanism is set for selling items, the buyer purchases a set $\itemsbought$ of the items that maximizes $v(\itemsbought) - p(\itemsbought)$, where $v(\itemsbought)$ is the desirability of $\itemsbought$ for the buyer and $p(\itemsbought)$ is the price that he pays. The revenue achieved by a mechanism is equal to
$\sum \mathbb{E}\big[p(\itemsbought)\big]$ where $V$ is randomly drawn from $D$.
The following terminology is used in \cite{babaioff2014simple} in order to compare the performance of different mechanisms. In this paper we use similar notations.
\begin{itemize}
\item \textbf{$\REV(D$):} Maximum possible revenue that can be achieved by any truthful mechanism.
\item \textbf{$\SREV(D)$:} The revenue that we get when selling items separately using Myerson's optimal mechanism for selling each item.
\item \textbf{$\BREV(D)$:} The revenue that we get when selling all items as a whole bundle using Myerson's optimal mechanism.
\end{itemize}
We refer to the expected value and variance of a one-dimensional distribution $D$ by $\VAL(D)$ and $\VAR(D)$ respectively.
We say an $n$-dimensional distribution $D$ of the desirabilities of a buyer is independent over the items if for every $a \neq b$, $v_{a}$ and $v_{b}$ are independent variables when $V = \langle v_{1},v_{2},\ldots,v_{n} \rangle$ is drawn from $D$.
Furthermore, we define the semi-independent distributions as follows.
\begin{definition}
Let $D$ be a distribution of valuations of a buyer over a set of items. We say $D$ is \textit{semi-independent} iff the valuations of every two different items are either always equal or completely independent. Moreover, we say two items $a$ and $b$ are similar in a semi-independent distribution $D$ if for every $V \sim D$ we have $v_{a} = v_{b}$.
\end{definition}
Moreover, we define the common base-value distributions as follows.
\begin{definition}
We say a distribution $D$ is common base-value, if there exist independent distributions $F_1,F_2,\ldots,F_n,B$ such that for $V = \langle v_{1},v_{2},\ldots,v_{n} \rangle \sim D$ and every $1 \leq j \leq n$, $v_{j} = f_j + b$ where $f_j$ comes from distribution $F_j$ and $b$ is drawn from $B$ which is equal for all items.
\end{definition}
A natural generalization of common base-value distributions are distributions in which the valuation of each item is determined by a linear combination of $k$ independent variables which are the same for all items. More precisely, we define the linear distributions as follows.
\begin{definition}
Let $D$ be a distribution of valuations of a buyer for $n$ items. We say $D$ is a \convexdistribution if there exist independent desirability distributions $F_1,F_2,\ldots,F_k$ and a $k \times n$ matrix $M$ with non-negative rational values such that $V = \langle v_{1},v_{2},\ldots,v_{n} \rangle \sim D$, can be written as $W \times M$ where $W = \langle w_1,w_2,\ldots,w_k \rangle$ is a vector such that $w_i$ is drawn from $F_i$.
\end{definition}
\section{The Core Decomposition Technique}\label{core-dep}
Most of the results in this area are mainly achieved by the core decomposition technique which was first introduced in  \cite{li2013revenue}. Using this technique we can bound the revenue of an optimal mechanism without taking into account the complexities of the revenue optimal mechanism. The underlying idea is to split distributions into two parts: the core and the tail. If for each realization of the values we were to know in advance for which items the valuations in the core part will be and for which items the valuations in the tail part will be, we would achieve at least the optimal revenue achievable without such information. This gives us an intuition which we can bound the optimal revenue by the total sum of the revenues of $2^n$ auctions where in each auction we know which valuation is in which part. The tricky part then would be to separate the items whose valuations are in the core part from the items whose valuations are in the tail and sum them up separately. We use the same notation which was used in  \cite{babaioff2014simple} for formalizing our arguments as follows.
\begin{itemize}
\item $D_i$: The distribution of desirabilities of the buyer for item $i$.
\item $D_{A}$: ($A$ is a subset of items): The distribution of desirabilities of the buyer for items in $A$.
\item $r_i$: The revenue that we get by selling item $i$ using Myerson's optimal mechanism.
\item $r$: The revenue we get by selling all of the items separately using Myerson's optimal mechanism which is equal to $\sum r_i$.
\item $t_i$: A real number separating the core from the tail for the distribution of item $i$. we say a valuation $v_i$ for item $i$ is in the core if $0 \leq v_i \leq r_it_i$ and is in the tail otherwise.
\item $p_i$: A real number equal to the probability that $v_i > r_it_i$ when $v_i$ is drawn from $D_i$. 
\item $p_A$: ($A$ is a subset of items): A real number equal to the probability that $\forall i \notin A, v_i \leq r_it_i$ and $\forall i \in A, v_i > r_it_i$.
\item $D_i^C$: A distribution of valuations of the $i$-th item that is equal to  $D_{i}$ conditioned on $v_i \leq r_it_i$.
\item $D_i^T$: A distribution of valuations of the $i$-th item for the buyer that is equal to  $D_{i}$ conditioned on $v_i > r_it_i$.
\item $D_A^C$: ($A$ is a subset of items): A distribution of valuations of the items in $[N]-A$ for the buyer that is equal to  $D_{[N]-A}$ conditioned on $\forall i \notin A, v_i\leq r_it_i$.
\item $D_A^T$: ($A$ is a subset of items): A distribution of valuations of the items in $A$ for the buyer that is equal to  $D_{A}$ conditioned on $\forall i \in A, v_i > r_it_i$.
\item $D^A$: A distribution of valuations for all items which is equal to $D$ conditioned on both $\forall i \notin A, v_i \leq r_it_i$ and $\forall i \in A, v_i> r_it_i$.
\end{itemize}
In Lemma \ref{lm1} we provide an upper bound for $p_i$. Next we bound $\REV(D_i^C)$ and $\REV(D)$ in Lemmas \ref{lm2} and \ref{lm3} and finally in Lemma \ref{core-de} which is known as Core Decomposition Lemma we prove an upper bound for $\REV(D)$. All these lemmas are proved in  \cite{babaioff2014simple} for the case of independent setting. 

\begin{lemma}\label{lm0}
For every $A \subset [N]$, if the valuation of items in $A$ are independent of items in $[N]-A$ then we have $\REV(D) \leq \REV(D_{A}) + \VAL(D_{[N]-A}).$
\end{lemma}
\begin{proof}
	Suppose for the sake of contradiction that $\REV(D) > \REV(D_{A}) + \VAL(D_{[N]-A})$, we show one can sell items of $A$ to obtain an expected revenue more than $\REV(D_{A})$ which contradicts with maximality of $\REV(D_A)$. To this end, we add items of $[N]-A$ (which are of no value to the buyer) and do the following:
	\begin{itemize}
		\item We draw a valuation for items in $[N]-A$ based on $D$.
		\item We sell all items with the optimal mechanism for selling items of $D$.
		\item Finally, the buyer can return each item that has bought from the set $[N]-A$ and get refunded by the auctioneer a value equal to what has been drawn for that item. Note that, since the buyer has no desirability for these items, it is in his best interest to return them.
	\end{itemize}
	Note that, we fake the desirabilities of the buyer for items in $[N]-A$ with the money that the auctioneer returns in the last step. Therefore, the behavior of the buyer is as if he had a value for those items as well. Since the money that the auctioneer returns to the buyer is at most $\VAL([N]-A)$ (in expectation), and we he achieves $\REV(D)$ (in expectation) at first, the expected revenue that we obtain is at least $\REV(D)-\VAL([N]-A)$ which is greater than $\REV(D_{A})$ and contradicts with the maximality of $\REV(D_{A})$.
\end{proof}

\begin{lemma}\label{lm1}
$p_i \leq \frac{1}{t_i}$.
\end{lemma}
\begin{proof}
	Suppose we run a second price auction with reserve price $t_ir_i$. Since the revenue achieved by this auction is equal to $p_it_ir_i$ and is at most $\REV(D_i) = r_i$ we have $p_i \leq \frac{1}{t_i}$.
\end{proof}

\begin{lemma}\label{lm2}
$\REV(D_i^C) \leq r_i$.
\end{lemma}
\begin{proof}
	This lemma follows from the fact that $D^C_i$ is stochastically dominated by $D_i$. Therefore $\REV(D_i) \geq \REV(D_i^C)$ and thus $\REV(D_i^C) \leq r_i$.
\end{proof}

\begin{lemma}\label{lm3}
$\REV(D_i^T) \leq r_i/p_i$.
\end{lemma}
\begin{proof}
	By definition, the probability that a random variable drawn from $D_i$ lies in the tail is equal to $p_i$, therefore $\REV(D^T_i)$ cannot be more than $p_ir_i$, since otherwise $\REV(D_i)$ would be more than $r_i$ which is a contradiction.
\end{proof}

\begin{lemma}\label{lm4}
$\REV(D) \leq \sum_A p_A \REV(D^A)$.
\end{lemma}
\begin{proof}
	Suppose the seller has a magical oracle that after the realization of desirabilities, it informs him for which items the valuation of the buyer lies in the tail and for which items it lies in the core. Let $A$ be the set of items whose values lie in the tail. By definition, the maximum possible revenue (in expectation) that the seller can achieve in this case is $\REV(D^A)$ and this happens with probability $p_A$, therefore having the magical oracle, the maximum expected revenue of the seller is $\sum_A p_A\REV(D^A)$. Since this oracle gives the seller some additional information, the optimal revenue that the seller can guarantee in this case is at least as much as $\REV(D)$ and hence
	$$\REV(D) \leq \sum_A p_A\REV(D^A).$$
\end{proof}

For independent setting we can apply Lemma \ref{lm0} to Lemma \ref{lm4} and finally with application of some algebraic inequalities come up with the following inequality
$$\REV(D) \leq \VAL(D_{\emptyset}^C)+\sum_A p_A \REV(D_A^T).$$
Unfortunately this does not hold for correlated settings since in Lemma \ref{lm0} we assume valuation of items of $A$ are independent of the items of $[N]-A$. 
Therefore, we need to slightly modify this lemma such that it becomes applicable to the correlated settings as well. Thus, we add the following restriction to the valuation of items: For each subset $A$ such that $p_A$ is non-zero, the valuation of items in $A$ are independent of items of $[N]-A$.
\begin{lemma}\label{core-de}
If for every $A$ with $p_A > 0$ the values of items in $A$ are drawn independent of the items in $[N]-A$ we have
$\REV(D) \leq \VAL(D_{\emptyset}^C)+\sum_A p_A \REV(D_A^T)$.
\end{lemma}
\begin{proof}
According to Lemma \ref{lm4} we have
\begin{equation}\label{55e}
\REV(D) \leq \sum_A p_A \REV(D^A).
\end{equation}
Since for every $A$ such that $p_A > 0$ we know the values of items in $A$ are drawn independent of items in $[N]-A$, we can apply Lemma \ref{lm0} to Inequality (\ref{55e}) and come up with the following inequality.
\begin{equation*}
\REV(D) \leq \sum_A p_A [\VAL(D^C_A) + \REV(D^T_A)].
\end{equation*}
Note that, $D^C_\emptyset$ is an upper bound for $\VAL(D^C_A)$ for all $A$. Therefore
\begin{equation*}
\REV(D) \leq \sum_A p_A [\VAL(D^C_\emptyset) + \REV(D^T_A)].
\end{equation*}
We rewrite the inequality to separate $\VAL(D^C_\emptyset)$ from $\REV(D^T_A)$.
\begin{equation*}
\REV(D) \leq \sum_A p_A \REV(D^T_A) +\sum_A p_A \VAL(D^C_\emptyset) .
\end{equation*}
Since $\sum p_A = 1$ 
\begin{equation*}
\REV(D) \leq \VAL(D^C_\emptyset) + \sum_A p_A \REV(D^T_A).
\end{equation*}
\end{proof}

\section{Semi-Independent distributions}\label{sid}
In this section we show the better of selling items separately and as a whole bundle is approximately optimal for the semi-independent correlations. To do so, we first show $k\cdot \SREV(D) \geq \REV(D)$ where we have $n$ items divided into $k$ types such that items of each type are similar. Next we leverage this lemma in order to prove $\max\{\SREV(D),\BREV(D)\}$ achieves a constant-factor approximation of the revenue of an optimal mechanism.
We start by stating the following lemma which is proved in  \cite{maskin1989optimal}.
\begin{lemma}\label{identicalre}
In an auction with one seller, one buyer, and multiple similar items we have $\REV(D) = \SREV(D)$.
\end{lemma}
\begin{proof}[Proof of Lemma \ref{identicalre}]
	Suppose we have $n$ similar items with valuation function $D$ for the buyer. By definition $\REV(D) \geq \SREV(D)$ since $\REV(D)$ is the maximal possible revenue that we can achieve. Therefore we need to show $\REV(D)$ cannot be more than $\SREV(D)$. Since all the items are similar, $\SREV(D) = n \REV(D_{i})$ for all $1 \leq i \leq n$. In the rest we show $\REV(D)$ cannot be more than $n$ times of $\REV(D_{i})$. We design the following mechanism for selling just one of the items:
	\begin{itemize}
		\item Pick an integer number $g$ between 1 and $n$ uniformly at random and keep it private.
		\item Use the optimal mechanism for selling $n$ similar items, except that the prices are divided over $n$.
		\item At the end, give the item to the buyer that has bought item number $g$ (if any), and take back all other sold items.
	\end{itemize}
	Note that, in the buyer's perspective both the prices and the expectation of the number of items they buy are divided by $n$, therefore they'll have the same behavior as before. Since prices are divided by $n$ the revenue we get by the above mechanism is exactly $\frac{\REV(D)}{n}$ which implies $\REV(D) \leq n\REV(D_{i})$ and completes the proof.
\end{proof}

We also need Lemma \ref{neisan} proved in \cite{hart2012approximate} and \cite{babaioff2014simple} which bounds the revenue when we have a sub-domain $S$ two independent value distributions $D$ and $D'$ over disjoint sets of items. Moreover we use Lemma \ref{ttghart} as an auxiliary lemma in the proof of Lemma \ref{basiclemma}.
\begin{lemma} \label{neisan}{\bf (``Marginal Mechanism on Sub-Domain \cite{hart2012approximate, babaioff2014simple}'')}
	Let $D$ and $D'$ be two independent distributions over disjoint sets of items. Let $S$ be a set of values of $D$ and $D'$ and $s$ be the probability that a sample of $D$ and $D'$ lies in $S$, i.e. $s=\text{Pr}[(v, v') \sim D\times D' \in S]$. $s\REV(D\times D'| (v, v') \in S) \leq s\VAL(D|(v, v') \in S)+\REV(D')$.
\end{lemma}
\begin{lemma}\label{ttghart}
In a single-seller mechanism with $m$ buyers and $n$ items with a semi-independent correlation between the items in which there are at most $k$ non-similar items we have $\REV(D) \leq mk\cdot \SREV(D)$.
\end{lemma}
\begin{proof}
	First we prove the case $m=1$. The proof is by induction on $k$. For $k=1$, all items are identical and by Lemma \ref{identicalre} $\REV(D)=\SREV(D)$. Now we prove the case in which we have $k$ non-similar types assuming the theorem holds for $k-1$. Consider a partition of $D$ into two parts $S_1$ and $S_2$ where in $S_1$, $v_1c_1 \geq c_iv_i$ for each $i$ and in $S_2$ there is at least one type $i$ such that $c_iv_i > c_1v_1$. Let $D^1$ and $D^2$ denote the valuations conditioned on $S_1$ and $S_2$, respectively, and let $p_1$ and $p_2$ denote the probability that the valuations lie in $D^1$ and $D^2$. 
	Since we do not lose revenue due to having extra information about the domain
	\begin{equation}\label{ineq:1}
	\REV(D) \leq p_1\REV(D^1)+p_2\REV(D^2).
	\end{equation}
	Thus we need to bound $p_1\REV(D^1)$ and $p_2\REV(D^2)$.
	Let $D_{-i}$ denote the distribution of valuations excluding the items of type $i$. 
	Using Lemma \ref{neisan}, $p_1\REV(D^1) \leq p_1\VAL(D^1_{-1}) + \REV(D_1)$ and $p_2\REV(D^2) \leq p_2\VAL(D^2_1)+\REV(D_{-1})$. Hence by Inequality (\ref{ineq:1}), 
	\begin{equation}\label{ineq:2}
	\REV(D) \leq p_1\VAL(D^1_{-1}) + \REV(D_1) + p_2\VAL(D^2_1)+\REV(D_{-1}).
	\end{equation}
	Now the goal is to bound four terms in Inequality (\ref{ineq:2}).
	For the first term consider the following truthful mechanism. Assume we only want to sell the items of type one. We take a sample $v \sim D$ and then sell all $c_1$ items of type one in a bundle with price $\max_{2\leq i \leq k} \{c_iv_i\}$. With probability $p_1$, $c_1v_1 \geq \max_{2\leq i \leq k} \{c_iv_i\}$ and hence the bundle would be sold. Thus with probability $p_1$, valuations lie in $D^1$ which means for each $i$ $v_1c_1 \geq c_iv_i$ and the revenue we get is $c_1v_1$, therefore
	\begin{equation}\label{ineq:3}
	p_1\VAL(D^1_{-1}) \leq k\REV(D_1).
	\end{equation}
	For the third term we provide another truthful mechanism which can sell all items except the items of type one. Take a sample $v \sim D$, put all items of the same type in the same bundles, except the items of type one. Hence we have $k-1$ bundles. Price all bundles equal to $c_1v_1$. With probability $p_2$ at least one bundle has a valuation greater than $c_1v_1$ and as a result would be sold and the revenue is more than $\VAL(D^2_1)$. Moreover by Lemma \ref{identicalre} in each bundle the maximum revenue is achieved by selling the items separately, thus
	\begin{equation}\label{ineq:4}
	p_2\VAL(D^2_1) \leq \SREV(D_{-1}).
	\end{equation}
	Moreover by induction hypothesis,
	\begin{equation}\label{ineq:5}
	\REV(D_{-1}) \leq k\SREV(D_{-1})).
	\end{equation}
	Summing up inequalities (\ref{ineq:3}), (\ref{ineq:4}), and (\ref{ineq:5}), $p_1\VAL(D^1_{-1}) + \REV(D_1) + p_2\VAL(D^2_1)+\REV(D_{-1}) \leq k\SREV(D_1) + \REV(D_1) + \SREV(D_{-1}) + k\SREV(D_{-1})$. Therefore, $\REV(D) \leq (k+1)\SREV(D)$, as desired.
	
	Now we prove that for any $m\geq 1$, $\REV(D) \leq mk\SREV(D)$. Note that any mechanism for $m$ buyers provides $m$ single buyer mechanisms and $\REV(D)=\sum_{i=1}^m \REV_i(D)$, where $\REV_i(D)$ is the revenue for $i$-th buyer. Thus $\max_i\REV_i(D) \geq \frac{1}{m}\REV(D)$ and as a result $\REV(D) \leq mk\SREV(D)$.
\end{proof}

Next, we show $\max\{\SREV(D),\BREV(D)\} \geq \frac{1}{6} \cdot \REV(D)$. The proof is very similar in spirit to the proof of \citeauthor{babaioff2014simple} for showing $\max\{\SREV(D),\BREV(D)\}$ achieves a constant approximation factor of the revenue optimal mechanism in independent setting \cite{babaioff2014simple}. In this proof, we first apply the core decomposition lemma with $t_i = r/(r_in_i)$ and break down the problem into two sub-problems. In the first sub-problem we show $\sum_A p_A \REV(D_A^T) \leq 2\SREV(D)$ and in the second sub-problem we prove $4 \max\{\SREV(D),\BREV(D)\} \geq \VAL(D_{\emptyset}^C)$. Having these two bounds together, we apply the core decomposition lemma to imply $\max\{\SREV(D),\BREV(D)\} \geq \frac{1}{6} \cdot \REV(D)$.
\begin{lemma}\label{basiclemma}
Let $D$ be a semi-independent distribution of valuations for $n$ items in single buyer setting. In this problem we have $\max\{\SREV(D),\BREV(D)\} \geq \frac{1}{6} \cdot \REV(D).$
\end{lemma}
\begin{proof}[Proof of Lemma \ref{basiclemma}]
	We use the core decomposition technique to prove this lemma. Let $n_i$ be the number of items that are similar to item $i$. We set $t_i = r/(r_i n_i)$ and then apply the Core Decomposition Lemma to prove a lower bound for $\max\{\SREV(D),\BREV(D)\}$. According to this lemma we have
	\begin{equation*}
	\REV(D) \leq \Big[\sum_A p_A \REV(D_A^T) \Big] + \Big[\VAL(D_{\emptyset}^C)\Big].
	\end{equation*}
	To prove the theorem, we first show $\sum_A p_A \REV(D_A^T) \leq 2\SREV(D)$ and next prove $\VAL(D_{\emptyset}^C) \leq 4\cdot \max\{\SREV(D),\BREV(D)\} $ which together imply 
	\begin{equation*}
	\REV(D) \leq \Big[\sum_A p_A \REV(D_A^T) \Big] + \Big[\VAL(D_{\emptyset}^C)\Big]
	\end{equation*}
	\begin{equation*}
	\leq 2\SREV(D) + 4 \max\{\SREV(D),\BREV(D)\}
	\end{equation*}
	\begin{equation*}
	\leq (2+ 4) \max\{\SREV(D),\BREV(D)\} \leq 6\max\{\SREV(D),\BREV(D)\}.
	\end{equation*}
	
	\begin{proposition}
		If we set $t_i = r/(r_i n_i)$ the following inequality holds in the single buyer setting.
		\begin{equation}
		\sum_A p_A \REV(D_A^T) \leq 2\SREV(D)
		\end{equation}
		where $D$ is a semi-independent valuation function for $n$ items.
	\end{proposition}
	\begin{proof}
		According to Lemma \ref{basiclemma} we have
		\begin{equation}
		\REV(D_A^T) \leq d_A \SREV(D_A^T) \leq d_A \big(\sum_{i \in A} \REV(D_i^T)\big) \leq d_A \big(\sum_{i \in A} \frac{r_i}{p_i}\big).
		\end{equation}
		Therefore, the following inequality holds.
		\begin{equation}\label{haha2}
		\sum_A p_A \REV(D_A^T) \leq \sum_A p_A d_A \big(\sum_{i \in A} \frac{r_i}{p_i}\big).
		\end{equation}
		where $d_A$ is the number of non-similar items in $A$. By rewriting Equation (\ref{haha2}) we get
		\begin{equation}\label{haha3}
		\sum_A p_A \REV(D_A^T) \leq \sum_{i=1}^n \frac{r_i}{p_i} \big(\sum_{A \ni i} p_A d_A\big) = \sum_{i=1}^n r_i \sum_{j=1}^n j \frac{1}{p_i} \big(\sum_{A \ni i \wedge d_A = j} p_A\big).
		\end{equation}
		Note that, $\sum_{j=1}^n j \frac{1}{p_i} \sum_{A \ni i \wedge d_A = j} p_A$ is the expected number of different items in the tail, conditioned on  item $i$ being in the tail. All of similar items lie in the tail together, and this probability is at most $\frac{1}{t_j} = \frac{n_jr_j}{r}$. Therefore, apart from $i$, the expected number of different sets of similar items in the tail is at most 1 and hence $\sum_{j=1}^n j \frac{1}{p_i} \sum_{A \ni i \wedge d_A = j} p_A < 2$. Therefore,
		\begin{equation*}
		\sum_A p_A \REV(D_A^T) \leq  \sum_{i=1}^n r_i \sum_{j=1}^n j \frac{1}{p_i} \big(\sum_{A \ni i \wedge d_A = j} p_A\big) \leq \sum_{i=1}^n 2r_i = 2\SREV(D).
		\end{equation*}
	\end{proof}
	
	Next, we show that $\max\{\SREV(D),\BREV(D)\} $ is at least $\frac{\VAL(D_{\emptyset}^C)}{4}$ which completes the proof. In the proof of this proposition, we use the following Lemma which has been proved by  \citeauthor{li2013revenue} in  \cite{li2013revenue}.
	
	\begin{lemma}\label{variance}
		Let $F$ be a one-dimensional distribution with optimal revenue at most $c$ supported on $[0, tc]$. Then $\VAR(F) \leq (2t - 1)c^2$.
	\end{lemma}
	
	\begin{proposition}
		For a single buyer in semi-independent setting we have 
		\begin{equation}
		4 \max\{\SREV(D),\BREV(D)\} \geq \VAL(D_{\emptyset}^C)
		\end{equation}
		where $t_i = r/(r_i n_i)$.
	\end{proposition}
	\begin{proof}
		Since $\SREV(D) = r$, the proof is trivial when $\VAL(D_{\emptyset}^C) \leq 4r$. Therefore, from now on we assume $\VAL(D_{\emptyset}^C) > 4r$. We show that $\VAR(D^C_{\emptyset}) \leq 2r^2$ and use this fact in order to show $\BREV(D)$ is a constant approximation of $\REV(D^C_{\emptyset})$. To this end, we formulate the variance of $D^C_{\emptyset}$ as follows:
		\begin{equation}
		\VAR(D^C_{\emptyset}) = \VAR(D^C_1+D^C_2+\ldots+D^C_N) =
		\sum_{i=1}^n\sum_{j=1}^n \COVAR(D^C_i,D^C_j)
		\end{equation}
		Note that $\COVAR(D^C_i,D^C_j) = \VAR(D^C_i)$ if items $i$ and $j$ are equal and $0$ otherwise. Therefore,
		\begin{equation}
		\VAR(D^C_{\emptyset}) = \sum_{i=1}^n \VAR(D^C_i) \times n_i
		\end{equation}
		Recall that Lemma \ref{variance} states that $\VAR(D^C_i) \leq 2rr_i/n_i$, and thus
		\begin{equation}
		\VAR(D^C_{\emptyset}) \leq  \sum_{i=1}^n \VAR(D^C_i) \times n_i \leq  \sum_{i=1}^n 2rr_i \leq 2r^2
		\end{equation}
		Since $\VAL(D^C_{\emptyset}) \geq 4r$ and $\VAR(D^C_{\emptyset}) \leq 2r^2$, we can apply the Chebyshev's Inequality to show
		\begin{equation}
		Pr\Big[ \sum v_i \leq \frac{2}{5} \VAL(D_{\emptyset}^C) \Big] \leq \frac{\VAR(D)}{(1-\frac{2}{5})^2 \VAL(D_{\emptyset}^C)^2} \leq \frac{2r^2}{(1-\frac{2}{5})^2 16r^2} \leq \frac{25}{72}
		\end{equation}
		This implies that the following pricing algorithm yields a revenue at least of $\frac{47\cdot2}{72\cdot5} \VAL(D_{\emptyset}^C)$: put a price equal to $\frac{2}{5} \VAL(D_{\emptyset}^C)$ on the whole set of items as a bundle. Since, $\BREV(D)$ is the best pricing mechanism for selling all items as a bundle we have
		\begin{equation*}
		\BREV(D) \geq \frac{47\cdot2}{72\cdot5} \VAL(D_{\emptyset}^C) \geq \frac{\VAL(D_{\emptyset}^C) }{4}
		\end{equation*}
	\end{proof}
		
\end{proof}


\section{Common Base-Value Distributions}
In this section we study the same problem with a common base-value distribution. Recall that in such distributions desirabilities of the buyer are of the form $v_{j} = f_{j}+b_i$ where $f_{j}$ is drawn from a known distribution $F_{j}$ and $b_i$ is the same for all items and is drawn from a known distribution $B$. Again, we show $\max\{\SREV,\BREV\}$ achieves a constant factor approximation of $\REV$ when we have only one buyer. Note that, this result answers an open question raised by \citeauthor{babaioff2014simple} in  \cite{babaioff2014simple}. 
\begin{theorem}\label{openquestion}
For an auction with one seller, one buyer, and a common base-value distribution of valuations we have
$\max\{\SREV(D),\BREV(D)\} \geq \frac{1}{12} \times \REV(D).$
\end{theorem}
\begin{proof}
Let $\I$ be an instance of the auction. We create an instance $\COR(\I)$ of an auction with $2n$ items such that the distribution of valuations is a semi-independent distribution $D'$ where $D'_i = F_i$ for $1 \leq i \leq n$ and $D'_i = B$ for $n+1 \leq i \leq 2n$. Moreover, the valuations of the items $n+1,n+2,\ldots,2n$ are always equal and all other valuations are independent. Thus, by the definition, $D'$ is a semi-independent distribution of valuations and by Lemma \ref{basiclemma} we have
\begin{equation}\label{firsteq}
\max\{\SREV(D'),\BREV(D')\} \geq \frac{1}{6} \times \REV(D').
\end{equation}
Since every mechanism for selling the items of $D$ can be mapped to a mechanism for selling the items of $D'$ where items $i$ and $n+i$ are considered as a single package containing both items, we have
\begin{equation}\label{secondeq}
\REV(D) \leq \REV(D').
\end{equation}
Moreover, since in the bundle mechanism we sell all of the items as a whole bundle, the revenue achieved by bundle mechanism  is the same in both auctions. Hence,
\begin{equation}\label{thirdeq}
\BREV(D) = \BREV(D').
\end{equation}
Note that, we can consider $\SREV(D)$ as a mechanism for selling items of $\COR(\I)$ such that items are packed into partitions of size 2 (item $i$ is packed with item $n+i$) and each partition is priced with Myerson's optimal mechanism. Since for every two independent distributions $F_i,F_{i+n}$ we have  $\SREV(F_i \times F_{n+i}) \leq 2\cdot \BREV(F_i \times F_{n+i})$ we can imply
\begin{equation}
\SREV(D) = \sum_{i=1}^n \BREV(F_i \times F_{n+i}) \geq \sum_{i=1}^n \frac{\SREV(F_i \times F_{i+n})}{2} = \frac{\SREV(D')}{2}.
\end{equation}
According to Inequalities \eqref{firsteq},\eqref{secondeq}, and \eqref{thirdeq} we have
\begin{equation*}
\max\{\SREV(D),\BREV(D)\} \geq \max\{\SREV(D')/2,\BREV(D')\} \geq
\end{equation*}
\begin{equation*}
\max\{\SREV(D'),\BREV(D')\}/2 \geq \REV(D')/12 \geq \REV(D)/12.
\end{equation*}
\end{proof}

\section{Linear Correlations}
A natural generalization of common base-value distributions is an extended correlation such that the valuation of each item for a buyer is a linear combination of his desirabilities for some features where the distribution of desirabilities for the features are independent and known in advance. More precisely, let $F_{1},F_{2},\ldots,F_{l}$ be $l$ independent distributions of desirabilities of features for the buyer and once each value $f_{j}$ is drawn from $F_{j}$, desirability of the buyer for $j$-th item is determined by $V_{f} \cdot M_{j}$ where $V_{f} = \langle f_{1},f_{2},\ldots,f_{l} \rangle$ and $M$ is an $n \times l$ matrix containing non-negative values.

Note that, a semi-independent distribution of valuations is a special case of linear correlation where we have $n+1$ features $F_{1},F_{2},\ldots,F_{n+1}$ and $M$ is a matrix such that $M_{a,b} = 1$ if either $a = b$ or $b = n+1$ and $M_{a,b} = 0$ otherwise. In this case, $F_{n+1}$ is the base value which is shared between all items and each of the other distributions is dedicated to a single item.

In this section we show the better of selling items separately or as a whole bundle achieves an $O(\log k)$ factor approximation of $\REV(D)$ when the distribution of valuations for all features are the same and the value of each item is determined by the value of at most $k$ features. To this end, we first consider an independent setting where the distribution of items are the same up to scaling and prove $\BREV$ is at least $O(\frac{\SREV}{\log n})$. Next, we leverage this lemma to show the main result of this section.
\subsection{Independent Setting}
In this part, we consider distributions which are similar to independent identical distributions, but their values are scaled by a constant factor. In particular, $D$ is a scaled distribution of $F$ if and only if for every $X$, $Pr_{u \sim D}[u=X]=Pr_{u \sim F}[u=\alpha X]$. We provide an upper bound for the ratio of the separate pricing revenue to the bundle pricing revenue for a set of items with independent scaled distributions. The following proposition shows this ratio is maximized when the value for each item $i$ is either $0$ or a constant number.
\begin{proposition}
	For every distribution $D =D_1 \times D_2 \times \ldots \times D_n$, where $D_i$'s are independent, there is a $D' =D'_1 \times D'_2 \times \ldots \times D'_n$, such that $\frac{\SREV(D')}{\BREV(D')}\geq \frac{\SREV(D)}{\BREV(D)}$ and for each $D'_i$ there is an $X_i$ and $p_i$ such that $Pr_{u\sim D'_i}[u=X_i]=p_i$ and $Pr_{u\sim D'_i}[u=0]=1-p_i$.
\end{proposition}
\begin{proof}
	For each $1 \leq i \leq n$ let $u_i$ denote the Myerson price for $D_i$. Let $p_i=Pr_{u\sim D_i}[u\geq u_i]$. Thus the revenue for selling $i$ separately is $p_iu_i$. Now Let $D'_i$ be a distribution which is $0$ with probability $1-p_i$ and $u_i$ with probability $p_i$. Thus $\SREV(D')=u_ip_i=\SREV(D)$. However since for each $i$, $D_i$ dominates $D'_i$, $\BREV(D')\leq \BREV(D)$. Therefore, $\frac{\SREV(D')}{\BREV(D')}\geq \frac{\SREV(D)}{\BREV(D)}$.
\end{proof}

Thus from now on, we assume each item $i$ has value $u_i=\alpha_i u$ with probability $p$ and  $0$ with probability $1-p$. Without loss of generality we can assume $u_1 \geq u_2 \geq \ldots \geq u_n$. In order to prove our bound we need to use the following theorem  and Lemma \ref{lemefari}.

\begin{theorem}[See \citeauthor{yao2014n}\cite{yao2014n}]\label{YaoLi}
	There exists a constant $c_0$ such that for every integer
	$n \geq 2$ and every $D =D_1 \times D_2 \times \ldots \times D_n$, where $D_i = F$ are independent
	and identical distributions, we have
	$c_0 \BREV(D) \geq \REV(D)$.
\end{theorem}

\begin{lemma}\label{lemefari}
	There exists a constant $c$ such that for every integer $1\leq j \leq n$ and every $D=D_1 \times D_2 \times \ldots \times D_n$, where $D_i=\alpha_iF$ are independent scaled distributions, we have $c\BREV(D) \geq jpu_j$.
\end{lemma}
\begin{proof}
	To show there exists a constant $c$ such that $c\BREV(D) \geq jpu_j$, first we consider another set of items with distribution $D'$ such that $\BREV(D) \geq \BREV(D')$. Then we show there is a constant $c$ such that $c \BREV(D') \geq jpu_j$.
	
	Let $D'=D_j^j$ i.e., a set of $j$ items with independent and identical distribution $D_j$. Note that for each $i<j$, $u_i \geq u_j$. Thus if $i\leq j$, $D_i$ dominates $D_j$. Moreover we are ignoring the other $n-j$ items. This implies $\BREV(D) \geq \BREV(D')$.
	
	Now by Theorem \ref{YaoLi}, there is a constant $c_0$ such that
	\begin{equation}\label{sina01}
	c_0 \BREV(D') \geq \REV(D').
	\end{equation}
	On the other hand, by selling the items separately the revenue for each item is $pu_j$. Hence
	\begin{equation}\label{sina02}
	\SREV(D') = jpu_j.
	\end{equation}
	Thus we can conclude,
	\begin{align*}
	c_0 \BREV(D) &\geq c_0 \BREV(D') &\text{By Inequality \eqref{sina01}}\\
	&\geq \REV(D')\\
	&\geq \SREV(D')	&\text{By Equation \eqref{sina02}}\\
	&=jpu_j.
	\end{align*}
\end{proof}

\begin{lemma}\label{lemeasli}
	There exists a constant $c'$ such that for every $D=D_1 \times D_2 \times \ldots \times D_n$, where $D_i=\alpha_iF$ are independent scaled distributions, we have $c'\BREV(D) \geq \frac{1}{\log(n)}\SREV(D)$.
\end{lemma}
\begin{proof}
	First we prove there exists an integer $1\leq j \leq n$ such that $jpu_j \geq \frac{1}{1+\ln(n)}\SREV(D)$. Then using Lemma \ref{lemefari} we obtain $c \BREV(D) \geq \frac{1}{1+\ln(n)}\SREV(D)$.
	
	Each item $i$ has value $0$ with probability $1-p$, and $u_i$ with probability $p$. Thus the optimal separate price for item $i$ is $u_i$ and the expected revenue for that is $pu_i$. Thus $\SREV(D)=p\sum_{i=1}^n{u_i}$
	Assume by contradiction for every $1 \leq j \leq n$, $jpu_j < \frac{1}{1+\ln(n)}\SREV(D)=\frac{1}{1+\ln(n)}p\sum_{i=1}^n{u_i}$. Simplifying the equation and moving $j$ to the right hand size, for every $1 \leq j \leq n$ we have
	\begin{equation}\label{sina03}
	u_j < \frac{1}{1+\ln(n)}\frac{1}{j}\sum_{i=1}^n{u_i}.
	\end{equation}
	Summing up Inequality \eqref{sina03} for all $j$ we have
	\begin{equation}
	\sum_{j=1}^n{u_j} < \frac{1}{1+\ln(n)}\sum_{j=1}^n{\frac{1}{j}}\sum_{i=1}^n{u_i}.
	\end{equation}
	This implies $1 < \frac{1}{1+\ln(n)}H_n$, which is a contradiction. Thus there is an integer $j$ such that $jpu_j \geq \frac{1}{1+\ln(n)}\SREV(D)$. Now by Lemma \ref{lemefari}, there is a constant $c$ such that $c\BREV(D) \geq jpu_j$. Thus $c' \BREV(D) \geq \frac{1}{\log(n)}\SREV(D)$.
\end{proof}

\citeauthor{babaioff2014simple} \cite{babaioff2014simple} show that for $n$ independent items and a single additive buyer the maximum of $\SREV$ and $\BREV$ is a constant fraction of the maximal revenue. Thus from Lemma \ref{lemeasli} we can conclude that $c' \BREV(D) \geq \frac{1}{1+\ln(n)}\REV(D)$.
\begin{corollary}
	There exists a constant $c'$ such that for every $D=D_1 \times D_2 \times \ldots \times D_n$, where $D_i=\alpha_iF$ are independent scaled distributions, we have $c'\BREV(D) \geq \frac{1}{\log(n)}\REV(D)$.
\end{corollary}
\subsection{Correlated Setting}

The following theorem shows an $O(\log k)$ approximation factor for $\max\{\SREV,\BREV\}$ when considering a linear correlation with i.i.d distributions for the features.
\begin{theorem}\label{ech}
Let $D$ be a distribution of valuations for one buyer in an auction such that the correlation between items is linear. If each row of $M$ has at most $k$ non-zero entries, then 
$$\max\{\SREV(D),\BREV(D)\} \geq \frac{\REV(D)}{c'' \log n}$$
where $c'' > 0$ is a constant real number.
\end{theorem}
\begin{proof}
	Since we can multiply the entries of the matrix by any integer number and divide the values of distribution by that number without violating any constraint of the setting, for simplicity, we assume all values of $M$ are integer numbers. Let $\I$ be an instance of our auction. We create an instance $\COR(\I)$ of an auction a with semi-independent distribution as follows: Let $n_i$ be the total sum of numbers in $i$-th column of $M$. For each feature we put a set of items in $\COR(\I)$ containing $n_i$ similar elements. Moreover, we consider every two items of different types to be independent. We refer to the distribution of items in $\COR(\I)$ with $D'$. Each mechanism of auction $\I$ can be mapped to a mechanism of auction $\COR(\I)$ by just partitioning items of $\COR(\I)$ into some packages, such that package $i$ has $M_{i,a}$ items from $a$-th type, and then treating each package as a single item. Therefore we have 
	\begin{equation}\label{asab1}
	\REV(D) \leq \REV(D').
	\end{equation}
	Moreover, bundle mechanism has the same revenue in both auctions since it sells all items as a whole package. Therefore the following equation holds.
	\begin{equation}\label{asab2}
	\BREV(D) = \BREV(D')
	\end{equation} 
To complete the proof we leverage Lemma \ref{lemeasli} to compare $\SREV(D)$ with $\SREV(D')$. In the following we show
\begin{equation}\label{khorp}
\SREV(D) \geq \frac{\SREV(D')}{c' \log k}.
\end{equation}
for a constant number $c' > 0$.
Note that selling items of auction $\I$ separately, can be thought of as selling items of $\COR(\I)$ in partitions with at most $k$ non-similar items. To compare $\SREV(D)$ with $\SREV(D')$, we only need to compare the revenue achieved by selling each item of $\I$ with the revenue achieved by selling its corresponding partition in $\COR(\I)$. To this end, we create an instance $\COR(\I)_i$ of an auction for each item $i$ of $\I$ which is consisted of all items in $\COR(\I)$ corresponding to item $i$ of $\I$ such that all similar items are considered as a single item having a value equal to the sum of the values of those items.  Let $L_i$ be a partition of such items. When selling items of $\I$ separately, our revenue is as if we sell all items of $L_i$ as a whole bundle. Therefore this gives us a total revenue of $\BREV(\COR(\I)_i)$. However, when we sell items of $\COR(\I)$ separately, the revenue we get from selling items of $L_i$ is exactly $\SREV(\COR(\I)_i)$. Note that, since all of the features have the same distribution of valuation, $\COR(\I)_i$ contains at most $k$ independent items and the distribution of valuations for all items are the same up to scaling. Therefore, according to Lemma \ref{lemeasli} $\SREV(L_i) \leq \BREV(L_i) c' \log k$ for a constant number $c' > 0$ and hence Inequality \eqref{khorp} holds. 

Next, we follow an approximation factor for $\max\{\SREV(D),\BREV(D)\}$ from Inequalities \eqref{khorp}, \eqref{asab1}, and \eqref{asab2}. By Inequalities \eqref{khorp} and \eqref{asab2} we have
\begin{equation}
\max\{\SREV(D),\BREV(D)\} \geq \frac{\max\{\SREV(D')\BREV(D')\}}{c' \log k }
\end{equation}
Moreover according to Inequality \eqref{asab1} $\REV(D) \leq \REV(D')$ holds and by Lemma \ref{basiclemma} we have $\max\{\SREV(D'),\BREV(D')\} \geq \frac{\REV(D')}{6}$ which yields $$\max\{\SREV(D),\BREV(D)\} \geq \frac{\REV(D)}{6 c' \log k}$$
Setting $c'' = 6c'$ the proof is complete.
\end{proof}

{
\bibliography{rev-abrv}
\bibliographystyle{plainnat}
}


\end{document}